%% file: main.tex
\documentclass[11pt]{article}

\usepackage{graphicx}
\usepackage{subfigure}
\usepackage{booktabs} % for professional tables
\usepackage{diagbox}
\usepackage{fullpage}

\input{glodef}

\input{locdef}

\def\withcolors{0}
\def\withnotes{0}

\ifnum\withcolors=1
\newcommand{\newer}[1]{{\color{RoyalPurple} {#1}}} % Newer stuff
\newcommand{\newest}[1]{{\color{orange} {#1}}} % Even even newer stuff
\newcommand{\newerest}[1]{\textcolor{purple}{#1}}
\else

\newcommand{\newer}[1]{{{#1}}}
\newcommand{\newest}[1]{{#1}}
\newcommand{\newerest}[1]{{#1}}
\fi
\usepackage[colorinlistoftodos,textsize=scriptsize]{todonotes}
\ifnum\withnotes=1
\else

\fi

\usepackage[colorlinks,citecolor=blue]{hyperref}

\title{Communication Complexity in Locally Private Distribution Estimation and Heavy Hitters}

\author {
Jayadev Acharya\thanks{Supported by NSF-CCF-CRII-1657471, NSF-CCF-1846300 (CAREER). To appear at International Conference on Machine Learning (ICML), 2019.} \\
Cornell University\\
\tt{acharya@cornell.edu}
\and
Ziteng Sun$^{*}$\\
Cornell University\\
\tt{zs335@cornell.edu}
}

\begin{document}
\maketitle
\begin{abstract}
\checked{We consider the problems of distribution estimation and heavy hitter (frequency) estimation under privacy and communication constraints. While these  constraints have been studied separately, optimal schemes for one are sub-optimal for the other. We propose a sample-optimal $\eps$-locally differentially private (LDP) scheme for distribution estimation, where each user communicates only one bit, and requires \textit{no public randomness}. We show that Hadamard Response, a recently proposed scheme for $\eps$-LDP distribution estimation is also utility-optimal for heavy hitter estimation. Finally, we show that unlike distribution estimation, without public randomness where only one bit suffices, any heavy hitter estimation algorithm that communicates $o(\min \{\log n, \log k\})$ bits from each user cannot be optimal.}   
\end{abstract}

\section{Introduction}
\input{sec-intro.tex}
%\input{sec-results.tex}

%\section{Preliminaries and Notations} \label{sec:prelim}
%\input{sec-preliminaries.tex}
%\input{related-prior.tex}

%\section{Prior Work and New Results}
%\input{sec-results.tex}

\section{One-Bit Private-coin LDP Distribution Estimation} \label{sec:onebit}
\input{sec-1bit_priv.tex}

\section{Lower Bound on Communication Complexity of Symmetric Schemes} \label{sec:lower_symmetric}
\input{sec-symmetric.tex}

\section{Hadamard Response is Optimal for Heavy Hitter Estimation}\label{sec:hr_heavyhitter}
\input{sec-hadamard-heavyhitter.tex}

\section{Communication Lower Bounds for Heavy Hitter Estimation}\label{sec:lower_heavyhitter}
\input{sec-constantbit.tex}

\section{Experiments}
\input{experiments.tex}

%\newpage
\section*{Acknowledgements}
The authors thank the anonymous reviewers for pointing out the result in~\cite{BassilyS15} about one-bit public-coin schemes, and other valuable suggestions, and Cl\'ement Canonne and Sourbh Bhadane for providing feedback on the manuscript.
\bibliography{masterref}
\bibliographystyle{IEEEtran}

%%%%%%%%%%%%%%%%%%%%%%%%%%%%%%%%%%%%%%%%%%%%%%%%%%%%%%%%%%%%%%%%%%%%%%%%%%%%%%%
%%%%%%%%%%%%%%%%%%%%%%%%%%%%%%%%%%%%%%%%%%%%%%%%%%%%%%%%%%%%%%%%%%%%%%%%%%%%%%%
\end{document}

%% file: glodef.tex
\usepackage{lmodern}
\usepackage{xspace}
\usepackage{amsfonts,amsmath,amssymb,amsthm, pbox}

\usepackage{multirow}

\usepackage{dsfont} %\indic

\usepackage[shortlabels]{enumitem}
\setitemize{noitemsep,topsep=0pt,parsep=0pt,partopsep=0pt}
\setenumerate{noitemsep,topsep=0pt,parsep=0pt,partopsep=0pt}

\makeatletter
\@ifundefined{theorem}{% Theorems (each with its own style)
  \theoremstyle{definition}
  \newtheorem{definition}{Definition}
  \theoremstyle{plain}
  \newtheorem{theorem}{Theorem}
  \newtheorem{corollary}{Corollary}
  \newtheorem{lemma}{Lemma}

  \theoremstyle{remark}

}{}
\makeatother

\newcommand{\ignore}[1]{}

\newcommand{\II}{\mathbb{I}} % Added by Theertha on April 16th 2013.
\newcommand{\EE}{\mathbb{E}}

\newcommand{\RR}{\mathbb{R}}

\newcommand{\expectation}[1]{\EE\left[#1\right]}
\newcommand{\variance}[1]{\textrm{Var}\left(#1\right)}

\def \cR     {{\cal R}}

\def \cX     {{\cal X}}
\def \cY     {{\cal Y}}

\def \ceil#1{{\lceil{#1}\rceil}}

\newcommand{\absv}[1]{\left|#1\right|}
\def \Paren#1{{\left({#1}\right)}}

\newcommand{\ed}{\stackrel{\text{def}}{=}}

\newcommand{\probof}[1]{\Pr\Paren{#1}}

\def\ignore#1{}

\newcommand{\bi}{\begin{itemize}}
\newcommand{\ei}{\end{itemize}}
\def\orpro{\mathop{\mathchoice
   {\vee\kern-.49em\raise.7ex\hbox{$\cdot$}\kern.4em}
   {\vee\kern-.45em\raise.63ex\hbox{$\cdot$}\kern.2em}
   {\vee\kern-.4em\raise.3ex\hbox{$\cdot$}\kern.1em}
   {\vee\kern-.35em\raise2.2ex\hbox{$\cdot$}\kern.1em}}\limits}

\def\andpro{\mathop{\mathchoice
 {\wedge\kern-.46em\lower.69ex\hbox{$\cdot$}\kern.3em}
 {\wedge\kern-.46em\lower.58ex\hbox{$\cdot$}\kern.25em}
 {\wedge\kern-.38em\lower.5ex\hbox{$\cdot$}\kern.1em}
 {\wedge\kern-.3em\lower.5ex\hbox{$\cdot$}\kern.1em}}\limits}

\def\simge{\mathrel{%
   \rlap{\raise 0.511ex \hbox{$>$}}{\lower 0.511ex \hbox{$\sim$}}}}

\def\simle{\mathrel{
   \rlap{\raise 0.511ex \hbox{$<$}}{\lower 0.511ex \hbox{$\sim$}}}}

%% file: locdef.tex
\newcommand{\ab}{k}
\newcommand{\ns}{n}

\newcommand{\p}{p}
\newcommand{\q}{q}

\newcommand{\Xon}{X^\ns}
\newcommand{\Yon}{Y^\ns}

\newcommand{\eps}{\varepsilon}

\newcommand{\norm}[2]{\lVert#1\rVert_{#2}}

\newcommand{\Mlt}[1]{N_{#1}}

\newcommand{\checked}[1]{\textcolor{black}{#1}}

%% file: sec-intro.tex
Inferring efficiently from data forms the core of modern data science. In many applications, being able to perform inference from available data is perhaps the most critical step. However, in several cases, these data samples contain sensitive information about the various users, who would like to protect their information from being leaked. For example, medical data may contain sensitive information about individuals that can be inferred without proper design of the collection scheme, a key issue highly publicized following the publications of~\cite{Sweeney02, HomerSRDTMPSNC08}.

Private data release and computation has been studied in various domains, such as statistics, machine learning, database theory, algorithm design, and cryptography (See e.g.,~\cite{Warner65, Dalenius77, DinurN03, WassermanZ10, WainwrightJD12, ChaudhuriMS11}). \emph{Differential Privacy (DP)}~\cite{DworkMNS06} has emerged as one of the most popular notions of privacy (see~\cite{DworkMNS06, WassermanZ10, BlumLR13, McSherryT07, KairouzOV17}, references therein, and  the recent book~\cite{DworkR14}). DP has been adopted by companies including Google, Apple and Microsoft~\cite{AppleDP17, ErlingssonPK14, DingKY17}. 

A popular privacy definition is \emph{local differential privacy (LDP)}, which was perhaps first proposed in~\cite{Warner65}, and more recently in~\cite{BeimelNO08, KasiviswanathanLNRS11}, where users do not trust the data collector, and privatize their data before releasing it. LDP is a stringent privacy constraint that requires noise to be added at each sample, and thus providing privacy to all users, even if the data collector is compromised. Often, LDP guarantees come at the cost of increased data requirement for various canonical inference tasks. 

{Communication constraints are increasingly becoming a bottleneck to more and more distributed inference problems. For example, in mobile devices, and small sensors with a limited power/limited uplink capacity, the communication budget can overshadow the local computations performed at each of them. This has led to a growing interest in understanding various inference tasks under limited communication, where the users do not have enough communication to even transmit their data~\cite{BravermanGMNW16, DaganS18}, and recent works have established optimal bounds, and algorithms for fundamental problems such as distribution estimation and hypothesis testing~\cite{DiakonikolasGLNOS17, HanOW18, AcharyaCT18, AcharyaCT18b}. Like privacy, these works show that communication constraints also increase the data requirements for vatious tasks.} 

%In this work, we will consider two related problems, discrete distribution estimation, and frequency/heavy hitter estimation. 

\subsection{Notations and Set-up}
We consider the following distributed setting. The underlying domain of interest is a known discrete set $\cX$ of size $\ab$. Without loss of generality, let $\cX = [k] :=\{1, \ldots, \ab\}$. \checked{There are $\ns$ users, and where user $i$ observes $X_i\in[\ab]$, and then sends a message $Y_i\in\cY$, the output domain, to the central server (data collector, referee) $\cR$, who upon observing the messages $Y^\ns:=Y_1, \ldots, Y_{\ns}$ wants to solve some pre-specified inference task. At user $i$, the process of generating message $Y_i$ from input $X_i$ can be characterized via a channel (a randomized mapping) $W_i:[\ab]\to\cY$, where $W_i(x, y)$ is the probability that $Y_i=y$ given that $X_i=x$.}

We now instantiate LDP, and communication constraints as special cases of this model. 

\smallskip
\noindent\textbf{1. Local Differential Privacy.}
A scheme is $\eps$-Locally Differentially Private (LDP), if $\forall x, x'\in\cX$, and $\forall y\in \cY$, 
\begin{align}
	\frac{W_i(x, y)}{W_i(x', y)} \le e^{\eps}, \ \ \ \forall i=1, \ldots, \ns.\label{eqn:dp-definition}
\end{align} 
Throughout the paper, we consider the high privacy regime, where $\eps=O(1)$.
 
\smallskip
\noindent\textbf{2. Communication Constraints.} Given a communication budget $\ell>0$, the scheme is \emph{$\ell$-bit communication limited} if $\cY=\{0,1\}^\ell$, and the output messages are at most $\ell$ bits long. 

We consider two inference tasks that $\cR$ wants to solve.
 
\smallskip
\noindent\textbf{1. Discrete Distribution Estimation.} Let $\Delta_k$ denote the collection of all distributions over the input domain $[\ab]$.  $\p\in\Delta_\ab$ is an unknown distribution. User $i$ observes $X_i$, which is an independent draw from $\p$. The referee's goal, upon observing the messages $Y_1, \ldots, Y_\ns$ is to estimate $\p$. The goal is to design the $W^\ns:=W_1,\ldots, W_\ns$'s (satisfying the appropriate constraints), and a $\hat{\p}:\cY^\ns\to\Delta_k$ to minimize the expected minimax $\ell_1$ risk:
\begin{align}
r(\ell_1, \Delta_k):=\min_{\hat\p}\min_{W^\ns}\max_{\p \in \Delta_\ab}\expectation{\norm{\p - \hat\p}{1}}.
\end{align}

When $W_1, \ldots, W_\ns$ satisfy~\eqref{eqn:dp-definition}, we denote $r(\ell_1, \Delta_k)$ by $r_{\texttt{DP}}(\ell_1, \Delta_k, \eps)$, and when $W_1, \ldots, W_\ns$ are communication limited by at most $\ell$ bits, it is denoted by $r_{\texttt{CL}}(\ell_1, \Delta_k, \ell)$
% and such that for any $\p\in\Delta_k$, over the randomness in $\Xon$, and $\Yon$, with probability at least $9/10$, $\dli{\p}{\hat\p}\le \alpha$. We denote

\medskip
\noindent\textbf{2. Frequency/Heavy Hitter Estimation.} Unlike distribution estimation, in this case there is no distributional assumption on the $X_i$'s (i.e., $X^\ns:=X_1,\ldots, X_\ns$ can be any element in $[\ab]^\ns$), and the goal is to estimate the empirical distribution of the symbols. In particular, for $x\in[\ab]$, let $\Mlt{x}$ be the number of appearances of $x$ in $\Xon$. The objective is to estimate the $\Mlt{x}/\ns$'s from the messages under LDP constraints. \newest{This problem has been studied under the $\ell_\infty$ norm objective, namely the goal is to design , and $\hat{p}$ to minimize
\begin{align}
\label{eqn:ell-inf}
\!\!\!r_{\texttt{DP}}(\ell_\infty,k, \eps):=\min_{\hat\p}\min_{W^\ns}\max_{\Xon}\EE\left[\max_{x}{\absv{\hat p(x)-\frac{\Mlt{x}}{n}}}\right],\!
\end{align}
where the expectation is over the randomness over messages induced by the channels, and the estimator $\hat p$.}

\checked{We will consider simultaneous message passing (SMP) communication protocols, where each user sends their message to the server at the same time. Within these, we study both protocols that have access to public randomness, and those that do not.} 

\medskip
\textbf{1. Private-coin Schemes:} \newest{In private-coin schemes, the players choose their channels $W_i$'s independently, without coordination. Formally, $U_1, \ldots, U_n$ are $\ns$ mutually independent random variables distributed across users $1,\ldots, \ns$ respectively. User $i$ chooses $W_i$ as a function of $U_i$. Referee $\cR$ knows the distribution of $U_i$, but not the instantiation of $U_i$ used to choose $W_i$.}

We first observe that under $\eps$-LDP or communication-constraints, private-coin schemes can be assumed to be \emph{deterministic}, namely the channels $W_1, \ldots, W_\ns$ are all fixed a priori. 
\begin{lemma}
\label{lem:deterministic} 
Private-coin schemes, and deterministic schemes are equivalent under both LDP and communication-limited constraints.
%Private-coin LDP schemes, and communication-limited schemes are equivalent to deterministic schemes. 
\end{lemma}
\begin{proof}
Note that the set of channels satisfying~\eqref{eqn:dp-definition}, and those with output at most $\ell$ bits are both convex. 
For any user $i$, let $\EE_{U_i}[W_i]$ denote the expected channel, over randomness in $U_i$. By convexity, $\EE_{U_i}[W_i]$ also satisfies the constraints, and can therefore be chosen as the deterministic channel whose input-output behavior is the same as choosing $U_i$, and then $W_i$.  	
\end{proof}
  %In this case, each user can user a different encoding function $f_i$ and each $f_i$ can be described by a channel matrix $W_i^{(k\times m)}$ where $W_i(r,s) = \probof{Y_i = s | X_i = r}$.

\medskip
\textbf{2. Public-coin Schemes:} {In public-coin protocols, the users and referee all have access to a common random variable $U$. The users select their channels as a function of $U$, namely $W_i = f_i(U)$. $\cR$ solves the inference task using messages $Y^\ns$, and $U$.}

\medskip
\textbf{3. Symmetric Schemes:} {These are schemes, where each user uses the same process to select their channel $W_i$'s. From Lemma~\ref{lem:deterministic}, we can conclude for any private-coin symmetric scheme for LDP and communication constraints, there is a deterministic $W$ such that $W_1=\ldots=W_\ns=W$.}%In particular, for private-coin symmetric schemes,  $U_1, \ldots, U_\ns$ are the i.i.d., and the function mapping the $U_i$'s to $W_i$'s are the same for all users.} 

\newer{\textbf{Motivation.} Beyond the theoretical interest in understanding the power of different kinds of schemes and randomness, we note some practical implications of this work. In particular, private-coin schemes are easier to implement than public-coin schemes, since they do not require additional communication from the server specifying the common random variable $U$. Even within private-coin schemes, symmetric schemes are easier to implement, since all users perform the same operation. We restrict to SMP schemes and do not consider the more general interactive schemes. These operate in rounds, and in each round some users send their messages. The players can choose their channels upon observing these previous messages~\cite{QinYYKXR16, DuchiJW13, KasiviswanathanLNRS11, SmithTU17}.~\cite{SmithTU17} show that interactive schemes can be much more powerful than non-interactive schemes.}

\checked{We consider both LDP distribution estimation and heavy hitter estimation and study the trade-offs between utility and communication for these problems.}

\subsection{Prior Work}
Distribution estimation is a classical task in the centralized setting, where $\cR$ observes the true samples $X^\ns$, and it is known that~\cite{DevroyeL01}
\begin{align}
\label{eqn:optimal-error}
r(\ell_1, \Delta_k) = \Theta\Paren{\sqrt{\frac{k}{n}}}.
\end{align}
%which shows a linear dependence of $n$ on $k$. 

\checked{Distribution estimation under $\eps$-LDP is also well studied in the past few years~\cite{DuchiWJ13, ErlingssonPK14, WangHWNXYLQ16, KairouzBR16, YeB18, AcharyaSZ18a, Bassily18}.~\cite{DuchiWJ13, KairouzBR16, YeB18, AcharyaSZ18a} have given private-coin, symmetric schemes which for $\eps=O(1)$ (our regime of interest) achieve}
\begin{align}
\label{eqn:optimal-error-ldp}
r_{\texttt{DP}}(\ell_1, \Delta_k, \eps) = \Theta\Paren{\sqrt{\frac{k^2}{n\eps^2}}}.
\end{align}
\checked{This $\ell_1$ risk is optimal over all protocols, even while allowing public-coins~\cite{DuchiJW13, YeB18, AcharyaCT18b}. Note that compared to the centralized setting, the risk is a factor of $\Theta(\sqrt{k}/\eps)$ higher which shows the significant drop in the utility under LDP.}

Distribution estimation under LDP constraints has only been formally studied using private-coin schemes. In terms of communication requirements (number of bits to describe $Y_i$'s)~\cite{DuchiWJ13, YeB18} require $\Omega(k)$ bits per user, and Hadamard Response (HR) of~\cite{AcharyaSZ18a} requires $\log k+2$ bits.

%\textcolor{red}{This comment above should probably go to the end of this section!! What do you think??} 
 
 Distribution estimation has also been studied recently under very low communication budget~\cite{DiakonikolasGLNOS17, HanOW18, AcharyaCT18, AcharyaCT18b}, where each user sends only $\ell<\log k$ bits to $\cR$. In particular, now it is established that by only using private-coin communication schemes,
 \begin{align}
\label{eqn:optimal-error-cl}
r_{\texttt{CL}}(\ell_1, \Delta_k, \ell) = \Theta\Paren{\sqrt{\frac{k^2}{n \min\{2^\ell, k\}}}}.
\end{align}	
Further, these results are tight even with public coins. 
Note that for $\ell=1$ (each user sends one bit),
 \begin{align}
\label{eqn:optimal-error-cl-onebit}
r_{\texttt{CL}}(\ell_1, \Delta_k, 1) = \Theta\Paren{\sqrt{\frac{k^2}{n}}}.
\end{align}	
\eqref{eqn:optimal-error-ldp} and~\eqref{eqn:optimal-error-cl-onebit} show the parallel between LDP, and communication constraints for $\ell_1$ risk of discrete distribution estimation. In particular, note that for both constraints there is an additional factor of $\Theta(\sqrt{k})$ blow-up in the $\ell_1$ risk over the centralized setting.

\checked{The problem of heavy hitter estimation under $\eps$-LDP has also received a lot of attention~\cite{HsuKR12, ErlingssonPK14, BassilyS15, QinYYKXR16, BassilyNST17, DingKY17, BunJU18}. The optimal $\ell_\infty$ (see~\eqref{eqn:ell-inf}) risk was established in~\cite{BassilyS15} as}
\begin{align}
r_{\texttt{DP}}(\ell_\infty,k, \eps):=\Theta\Paren{\frac1\eps\sqrt{\frac{\log k}{\ns}}}.
\end{align}
The state of the art research focus on improved computation time, and reduced communication from the users at the expense of public randomness from the referee~\cite{BassilyS15, BassilyNST17, BunJU18}.
\checked{In particular, these works propose algorithms that require only $O(1)$ communication from each user, and more interestingly, are  able to achieve a running time (at the referee) that is almost linear in $n$, and only logarithmic in $\ab$. We also note that the algorithm in~\cite{BassilyNST17} can in fact be simulated at the  users who can then transmit it, causing an increased communication cost from the users.} %We remark that the primary concern of these papers was to design both computation and communication efficient schemes.

The complementary problem of hypothesis testing has also been studied in the locally differentially private setting in~\cite{Sheffet17, GaboardiR18, AcharyaCFT19}, where the goal is to test whether the data samples are generated from one class of distributions or another. 

\subsection{Our Results and Techniques}

\input{sec-results}

We start by stating informally a remarkable result of~\cite{BassilyS15} that can help provide better context for some of the prior work and our new results. 
\begin{lemma}[\cite{BassilyS15}]
\label{lem:bassily-smith}
Any private-coin scheme with arbitrary communication requirements can be converted into a public-coin scheme that requires only one bit of communication from each user with almost no loss in performance.
\end{lemma} 
This result implies that with public randomness there exist schemes for both distribution estimation and heavy hitter estimation with one bit communication from each user. In this context, the main contribution of~\cite{BunJU18, BassilyNST17} is to design schemes that are computationally efficient, as discussed earlier.  

In this paper, in some sense, we study the converse question.
\begin{center}
\noindent\fbox{%
	\parbox{5.5in}{%
		\textit{Is public randomness necessary to reduce communication from users under LDP?
}
	}
}
\end{center}
One of our main results is that for the two related problems of distribution estimation and heavy hitter estimation, the answer to this question is different. For LDP distribution estimation, we design a scheme with optimal $\ell_1$ risk that needs no public randomness, and transmits only one bit of communication per player. We also show any private-coin scheme with optimal risk for heavy hitter estimation must communicate $\omega(1)$ bits per user. We now describe our results and techniques. 

\medskip
\noindent\textbf{Distribution Estimation.} To build toward constant communication private-coin distribution estimation schemes, recall that the known optimal distribution estimation schemes~\cite{DuchiJW13, KairouzBR16, YeB18, AcharyaSZ18a} are all symmetric, and require $\Omega(\log k)$ communication bits per message. %Further, the 1-bit communication scheme that achieves~\eqref{eqn:optimal-error-cl-onebit} does not preserve any privacy ($\eps=\infty$), and is also asymmetric. 
%Our paper is motivated by the following question:

Our first result in Theorem~\ref{thm:dist_sym} (Section~\ref{sec:lower_symmetric}) states that any private-coin scheme that is both symmetric, and where each player transmits fewer than $\log k$ bits has minimax $\ell_1$ risk equal to 1. This shows that HR  has optimal communication (up to $\pm 2$ bits) among all symmetric schemes. More importantly, this implies that any private-coin scheme that aims to communicate fewer bits must be asymmetric, namely the users cannot all have the same $W_i$'s.

We then design an asymmetric private-coin LDP distribution estimation scheme with optimal $\ell_1$ risk and where each user transmits only one bit. The precise result is given in Theorem~\ref{thm:1bit_performance} (Section~\ref{sec:onebit}).

Our scheme does the following. User $i$ is assigned a subset $B$ (can be different for different users) of the domain $[k]$ deterministically (hence private-coin). Upon observing $X_i$, they would like to send $\II\{X_i\in B\}$. To do so with $\eps$-LDP, they flip the random variable with probability $1/(e^\eps+1)$ and and send the result as $Y_i$.  We will choose subsets $B$'s defined by a Hadamard matrix of appropriate size, which not only achieves optimal $\ell_1$ risk, but also allow the server to compute $\hat p$ in nearly-linear time $\tilde{O}(k+n)$. 

Our scheme, while inspired by Hadamard Response, differs from it in the following sense. As described in Section~\ref{sec:hr_heavyhitter}, HR considers a Hadamard matrix where the rows are indexed by the input domain $[k]$. Upon observing a symbol $X_i$, user $i$ considers the $X_i$th row of the Hadamard matrix, and transmits an index of the columns based on whether that entry is a 1 or not in the matrix. This requires about $\log k$ bits. On the other hand, our one-bit scheme assigns to each user, a column of the Hadamard matrix, and the locations in that column that have a $1$ correspond to the subset assigned to that user.

\medskip
\noindent\textbf{Heavy Hitter Estimation.} All known optimal algorithms  for heavy hitter estimation described in the previous section use public randomness. In Theorem~\ref{thm:hadamard_heavy}, we show that HR (symmetric, and no public randomness, and $\log k +2$ bits of communication per user) has the optimal $\ell_\infty$ risk for heavy hitter estimation. 

However, we remark that the computation requirements of HR is $O(k\log k+n)$, which can be much worse than the guarantees in~\cite{BassilyNST17} for $k\gg \ns$.

\checked{Finally, in  Theorem~\ref{thm:heavy_accuracy} we show that unlike for distribution estimation, any private-coin scheme for heavy hitter estimation requires large communication. In particular, we show that optimal private-coin heavy-hitter estimation requires $\Omega(\min\{\log k, \log n\})$ bits of communication per user.} %The result also holds without privacy constraints.

For a complete summary of results, see Table~\ref{tab:dist_learning} and~\ref{tab:heavy-hitters}. \newerest{In each table, the problem becomes easier as we go down in rows and go right in columns.}
% and Table~\ref{tab:freq}.

\paragraph{Organization.}
\checked{In Section~\ref{sec:onebit}, we provide an optimal one-bit private-coin distribution estimation scheme. In Section~\ref{sec:lower_symmetric},  we show that any symmetric private-coin scheme must transmit $\log k$ bits per user. In Section~\ref{sec:hr_heavyhitter}, we prove the optimality of Hadamard Response for heavy hitter estimation and finally in Section~\ref{sec:lower_heavyhitter}, we show that without public randomness, heavy hitter estimation requires large communication.}

%% file: sec-results.tex
%\begin{table*} 
%	\centering
%	\begin{tabular} {| c | c | c | c|}\hline 
%		\backslashbox{Communication}{Randomness} & Symmetric & Private Randomness &  Public Randomness \\ \hline
%		$O(1)$ bits & $\Omega(1)$ (Theorem~\ref{thm:dist_sym})& $\Theta(\sqrt{\frac{k^2}{n\eps^2}})$ (Theorem~\ref{thm:1bit_performance})& $\Theta(\sqrt{\frac{k^2}{n\eps^2}})$ \\ \hline
%		$O(\log k)$ bits & $\Theta(\sqrt{\frac{k^2}{n\eps^2}})$~\cite{AcharyaSZ18a} & $\Theta(\sqrt{\frac{k^2}{n\eps^2}})$ &  $\Theta(\sqrt{\frac{k^2}{n\eps^2}})$\\ \hline
%	\end{tabular}
%	\caption{$\ell_1$ risk for distribution learning under different communication budget and available randomness.}
%	\label{tab:dist_learning}
%\end{table*}
\begin{table*}
	\centering
	\begin{tabular} {| c | c | c |}\hline 
		\backslashbox{Randomness}{Communication}& $O(1)$ bits& $O\left(\log k\right)$ bits  \\ \hline
		Symmetric, Private Randomness& $\Omega(1)$ (Theorem~\ref{thm:dist_sym}) & $\Theta\Paren{\sqrt{\frac{k^2}{n\eps^2}}}$~\cite{AcharyaSZ18a} \\ \hline
		Private Randomness & $\Theta(\sqrt{\frac{k^2}{n\eps^2}})$ (Corollary~\ref{thm:onebit_l1})& $\Theta\Paren{\sqrt{\frac{k^2}{n\eps^2}}}$  \\ \hline
		Public Randomness &$\Theta\Paren{\sqrt{\frac{k^2}{n\eps^2}}}$ & $\Theta\Paren{\sqrt{\frac{k^2}{n\eps^2}}}$\\ \hline
	\end{tabular}
	\caption{$\ell_1$ risk for distribution estimation under different communication budget and randomness.}
	\label{tab:dist_learning}
\end{table*}

%\begin{table*} 
%	\centering
%	\begin{tabular} {| c | c | c | c|}\hline 
%		\backslashbox{Communication}{Randomness} & Symmetric & Private Randomness &  Public Randomness \\ \hline
%		$O(1)$ bits & $\Omega(1)$ & $\Omega(1)$ (Theorem~\ref{thm:heavy_accuracy}) & $\Theta(\sqrt{\frac{\log k}{n\eps^2}})$~\cite{BassilyS15} \\ \hline
%		$O(\log k)$ bits & $\Theta(\sqrt{\frac{\log k}{n\eps^2}})$ (Theorem~\ref{thm:hadamard_heavy}) & $\Theta(\sqrt{\frac{\log k}{n\eps^2}})$ & $\Theta(\sqrt{\frac{\log k}{n\eps^2}})$ \\ \hline
%	\end{tabular}
%	\caption{$\ell_\infty$ risk for frequency estimation under different communication budget and available randomness.}
%	\label{tab:dist_learning}
%\end{table*}

\begin{table*}
	\centering
	\begin{tabular} {| c | c | c |}\hline 
		\backslashbox{Randomness}{Communication}& $O(1)$ bits& $O(\log k)$ bits  \\ \hline
		Symmetric, Private Randomness& $\Omega(1)$ & $\Theta\Paren{\sqrt{\frac{\log k}{n\eps^2}}}$ (Theorem~\ref{thm:hadamard_heavy}) \\ \hline
		Private Randomness & $\Omega(1)$ (Theorem~\ref{thm:heavy_accuracy}) & $\Theta\Paren{\sqrt{\frac{\log k}{n\eps^2}}}$ \\ \hline
		Public Randomness & $\Theta\Paren{\sqrt{\frac{\log k}{n\eps^2}}}$~\cite{BassilyS15} &  $\Theta\Paren{\sqrt{\frac{\log k}{n\eps^2}}}$ \\ \hline
	\end{tabular}
	\caption{$\ell_\infty$ risk for frequency estimation under different communication budget and schemes.}
	\label{tab:heavy-hitters}
\end{table*}

%\subsection{One bit distribution learning without public randomness}\label{sec:one-bit}
%\subsection{Communication lower bound on symmetric schemes}\label{sec:symmetric-lb}
%\subsection{Optimal Heavy Hitters with Hadamard Response}\label{sec:hadamard}
%\subsection{One Bit Heavy Hitter is Impossible}\label{sec:lb-one-bit}
%\subsection{Heavy Hitters with Fewer than log k bits is Impossible}

%% file: sec-1bit_priv.tex
\newest{We propose a deterministic scheme, namely the $W_i$'s are fixed apriori, for LDP distribution estimation that has the optimal $\ell_1$ risk, and where the output of each $W_i$ is binary, i.e., one bit of communication per user. The approach is the following. Each user is assigned to a deterministic set $B\subset[k]$. Upon observing a sample $X\sim\p$, they output $Y\in\{0,1\}$, according to the following distribution}
\begin{equation} \label{eqn:response}
\probof{Y = 1} = \begin{cases}
\frac{e^\eps}{e^\eps +1}, \text{if } X \in B, \\
\frac{1}{e^\eps +1}, \text{otherwise.}
\end{cases}
\end{equation}
In other words, each user sends the indicator of whether their input belongs to a particular subset of the domain. The choice of the subsets is inspired by the Hadamard Response (HR) scheme described in~\cite{AcharyaSZ18a}. A brief introduction of HR can be found in Section~\ref{sec:hr_heavyhitter} where we show that HR is utility-optimal for heavy hitter estimation.

Recall Sylvester's construction of Hadamard matrices.
%\subsection{Hadamard matrix} \label{sec:hadamard}
%We will use Hadamard matrices whose size is a power of 2. For $m$ that is a power of 2, $H_{m}$ is the Hadamard matrix of size $m\times m$, defined by Sylvester's construction: 
%\begin{definition}
%	\label{def:hadamard}
%	Let $H_1\ed[1]$, and for $m= 2^j$, for $j\ge1$, 
%	\[H_m \ed \begin{bmatrix} H_{m/2} & H_{m/2} \\ H_{m/2} & -H_{m/2} \end{bmatrix}.\]
%	%All the entries of $H_m$ are in $\{-1,1\}$.
%\end{definition}
\begin{definition}\label{def:hadamard}
	Let $H_1\ed[1]$, and for $m= 2^j$, for $j\ge1$, 
	\[H_m \ed \begin{bmatrix} H_{m/2} & H_{m/2} \\ H_{m/2} & -H_{m/2} \end{bmatrix}.\]
	%All the entries of $H_m$ are in $\{-1,1\}$.
\end{definition}

\newest{Let $K = 2^{\ceil{\log_2 (k+1)}}$ be the smallest power of 2 larger than $k$. Let $H_K$ be the $K\times K$ Hadamard matrix. For simplicity of working with $H_K$, we assume that the underlying distribution is over $[K]$ by appending $\p$ with zeros, giving $\p_K=(\p(1), \ldots, \p(k), 0, \ldots, 0)$. For $i=1, \ldots, K$, let $B_i$ be the set of all $x\in[K]$, such that $H_K(x,i)=1$, namely the row indices that have `1' in the $i$th column. We associate the subsets for each user as follows. We deterministically divide the $\ns$ users numbered $1,\ldots, \ns$ into $K$ subsets $S_1, S_2, \ldots, S_{K}$, such that}
\[
	S_i := \{ j \in [n] | j \equiv i \pmod{K}\}.
\]
For each user $j$, let $i_j\in[K]$ be the index such that $j\in S_{i_j}$. The $j$th user then sends its binary output $Y_i$ according to the distribution in~\eqref{eqn:response}, with $B=B_{i_j}$, and $X=X_j$. 

For any $i=1, \ldots, K$, the users in $S_i$ have the same output distribution. Let $s_i$ be the probability $Y_j=1$ for $j\in S_i$. Let $\p(B_i) = \probof{X \in B_i | X \sim \p}$. Note that
\begin{align}
\label{eqn:pB2pY}
s_i &= \p(B_i)\cdot \frac{e^{\eps}}{e^\eps +1}+\Paren{1-\p(B_i)} \frac{1}{e^\eps +1} \nonumber \\
&=\frac{1}{e^\eps +1}+ \p(B_i)\cdot \frac{e^{\eps}-1}{e^\eps +1}.
\end{align}

Let $\p_B:=(\p(B_1), \p(B_2), \ldots, \p(B_{K}))$. Then we obtain 
\begin{align}
\mathbf{s} :=(s_1,\ldots, s_{K}) = \frac{1}{e^\eps +1}\mathbf{1}_K+  \frac{e^{\eps}-1}{e^\eps +1}\p_B.
\end{align}

This relates $\p(B_i)$ with $s_i$, and now we relate $\p(x)$ with $\p(B_i)$'s. Recall that $B_1=[K]$, the entire set. Since $B_i$'s are defined by the rows of Hadamard matrix, we obtain the following~\cite{AcharyaSZ18a},
\begin{align}
\label{eqn:hadamard}
	\p_B = \frac{H_K \cdot \p_K + 1_K}{2}.
\end{align}

We can now relate the results and describe our estimate. 
\begin{enumerate}
\item 
	Use an empirical estimate $\widehat{\mathbf{s}}$ for $\mathbf{s}$ as
	\begin{equation} \label{eqn:pY_est}
	\widehat{s}_i := \frac{1}{|S_i|} \sum_{j \in S_i} Y_j.
	\end{equation}
%	which is an unbiased empirical estimate.
\item Motivated by~\eqref{eqn:pB2pY} estimate $\p_B$ as
	\begin{equation} \label{eqn:pB_est}
		\widehat{\p_B} =  \frac{e^\eps +1} {e^\eps - 1} \Paren{\widehat{\mathbf{{s}}}  - \frac{\mathbf{1}_K}{e^\eps +1}}.
	\end{equation}
\item Estimate for the original distribution using~\eqref{eqn:hadamard} as
	\begin{align} \label{eqn:reverse_hadamard}
		\widehat{\p_K}  :=  H_K^{-1} \cdot \Paren{2\widehat{\p_B} - 1_K} = \frac{1}{K} H_K \cdot \Paren{2\widehat{\p_B} - 1_K}.
	\end{align}
%	where $\widehat{p_C} = \Paren{ 1, \widehat{\p_C(1)}, \widehat{\p_C(2)}, ... ,\widehat{\p_C(K-1)}}$
	\item Output $\widehat{\p}$, the projection of the first $k$ coordinates of the $K$ dimensional $\widehat{p_K}$ on the simplex $\triangle_k$.
\end{enumerate}

\begin{theorem}\label{thm:1bit_performance}
	Let $\widehat{\p}$  be the output of the scheme above when the underlying distribution is $\p$. Then,
	\begin{align}
		\expectation{\norm{\widehat{\p} - \p}{2}^2} \le\min \left\{ \frac{2k(e^\eps + 1)^2}{n(e^\eps - 1)^2}, 8\sqrt{\frac{(e^\eps + 1)^2\log{k}}{n(e^\eps - 1)^2}}\right\}.\nonumber
	\end{align}
\end{theorem}

\begin{proof}
First note that $\widehat{\mathbf{s}}$ is an unbiased estimator of $\mathbf{s}$, and~\eqref{eqn:hadamard},~\eqref{eqn:pB2pY} and~\eqref{eqn:pY_est}, are all linear. Therefore, $\widehat{\p_K}$ is an unbiased estimator of $\p_K$. Hence,
	\begin{align}
		\expectation{\norm{\widehat{\p_K}(1: k) - \p_K(1: k)}{2}^2} = \sum_{x = 1}^{k} \variance{\widehat{\p_K(x)}}. \nonumber 
	\end{align}
From~\eqref{eqn:reverse_hadamard}, $\widehat{\p_K(x)}$ is a weighted sum of $\{(2\widehat{\p_B(i)} - 1)\}_{i=1}^{K}$ with coefficients either $+\frac{1}{K}$ or $-\frac{1}{K}$. Hence $\forall x \in [K]$,
	\begin{align}
		\variance{\widehat{\p_K(x)}} & \le \frac{4}{K^2} \sum_{y = 1}^{K} \variance{\widehat{\p_B(y)} } = \frac{4}{K^2} \Paren{\frac{e^\eps +1} {e^\eps - 1} }^2  \sum_{i = 1}^{K} \variance{\widehat{s_i} }. \nonumber 
	\end{align}
	By~\eqref{eqn:pY_est}, $\widehat{s_i}$ is an average of $|S_i|$ independent Bernoulli random variables with $|S_i| \ge \frac{n}{2K}$, hence $\forall i \in [K]$,
	\begin{align}
		\variance{\widehat{s_y} } \le \frac{K}{2n}. \nonumber 
	\end{align}
	Combining these, we get: $\forall x \in [K]$,
	\begin{align}
		\variance{\widehat{\p_K(x)}} = \frac{4}{K^2} \Paren{\frac{e^\eps +1} {e^\eps - 1} }^2  \sum_{i = 1}^{K} \variance{\widehat{s_i} }  
		\le \frac{2(e^\eps +1)^2} {n(e^\eps - 1)^2} . \nonumber 
	\end{align}
%\begin{align}
%	\variance{\widehat{\p_K(x)}} & = \frac{4}{K^2} \Paren{\frac{e^\eps +1} {e^\eps - 1} }^2  \sum_{y = 1}^{K} \variance{\widehat{s_y} }  \nonumber \\
%	&\le \Paren{\frac{e^\eps +1} {e^\eps - 1} }^2 \frac{2}{n}. \nonumber 
%\end{align}
	Then the final estimate $\widehat{\p}$ is the projection of $\p_K$ on the first $k$ coordinates onto the simplex $\triangle_k$. Since $\triangle_k$ is convex,
	\begin{align}
		\expectation{\norm{\widehat{\p} - \p}{2}^2} & \le  \expectation{\norm{\widehat{\p_K}(1: k) - \p_K(1: k)}{2}^2}  \nonumber \\
		& \le \sum_{i = 1}^k \Paren{\frac{e^\eps +1} {e^\eps - 1} }^2 \frac{2}{n} \le \frac{2k(e^\eps + 1)^2}{n(e^\eps - 1)^2}. \nonumber 
	\end{align}
Moreover, we have the following lemma.
\begin{lemma} \label{lem:projection}
	(Corollary 2.3~\cite{Bassily18}) Let $L\subset R^d$ be a symmetric convex body of $k$ vertices $\{\mathbf{a_j}\}_{j = 1}^k$, and let $\mathbf{y} \in L$ and $\mathbf{\bar{y}} = \mathbf{y}+\mathbf{z} $ for some $\mathbf{z} \in R^d$. Let $\mathbf{\hat{y}} = \arg\min_{\mathbf{w}\in L} \norm{\mathbf{w} - \mathbf{\bar{y}}}{2}^2$. Then, we must have:
	\begin{align}
		\norm{\mathbf{y} - \mathbf{\hat{y}}}{2}^2 \le 4 \max_{j\in[k]}\{\langle\mathbf{z}, \mathbf{a_j}\rangle\}
	\end{align}
\end{lemma}
	
	Notice that according to~\eqref{eqn:pB_est},~\eqref{eqn:pY_est} and~\eqref{eqn:pB2pY}, $ \{ \widehat{\p_B}(i) - \p_B(i) \}_{i=1}^K$ are empirical averages of independent zero mean Bernoulli random variables scaled by constant $\frac{e^\eps +1} {e^\eps - 1}$ and they are mutually independent. Hence, they are sub-Gaussian with variance proxy $\frac{K}{2n}\Paren{\frac{e^\eps +1} {e^\eps - 1}}^2$.
	
	Additionally, by~\eqref{eqn:pB2pY} and~\eqref{eqn:reverse_hadamard}, we know each of $\{\widehat{\p_K(x)} - \p_K(x)\}$ is a linear combination of $\{ \widehat{\p_B}(i) - \p_B(i) \}_{i=1}^K$ with coefficient either $+ \frac{2}{K}$ or $-\frac{2}{K}$. Hence $\{\widehat{\p_K(x)} - \p_K(x)\}$'s are also sub-Gaussian with variance proxy $\frac{2}{n}(\frac{e^\eps +1} {e^\eps - 1})^2$ (see Corollary 1.7~\cite{Rigollet15}).
	
	Hence using Lemma~\ref{lem:projection}, we have:
	\begin{align}
		\expectation{\norm{\widehat{\p} - \p}{2}^2} & \le 4 \expectation{\max_{x=1}^{k} |\widehat{\p_K(x)} - \p_K(x)|}   \le 8\sqrt{\frac{(e^\eps + 1)^2\log{k}}{n(e^\eps - 1)^2}}. \nonumber 
	\end{align}

	The last step is due to a well-known bound on expectation of maximum of sub-Gaussian random variables (see Theorem 1.16~\cite{Rigollet15}).
\end{proof}

\begin{corollary} \label{thm:onebit_l1}
	Let $\widehat{\p}$  be the estimate given by the scheme described above. Then for any input $\p$, 
	\begin{align}
		\expectation{\norm{\widehat{\p} - \p}{1}} \le \sqrt{\frac{2k^2(e^\eps + 1)^2}{n(e^\eps - 1)^2}}.\nonumber
	\end{align}
\end{corollary}

\begin{proof}
	By Cauchy-Schwarz inequality, 
	\[\expectation{\norm{\widehat{\p} - \p}{1}} \le \sqrt{k \expectation{\norm{\widehat{\p} - \p}{2}^2}}.
	\]
	 Plugging in Theorem~\ref{thm:1bit_performance} gives the bound. 
%	\begin{align}
%		\expectation{\norm{\widehat{\p} - \p}{1}} \le \sqrt{k \expectation{\norm{\widehat{\p} - \p}{2}^2}} \le \sqrt{\frac{k^2(e^\eps + 1)^2}{2n(e^\eps - 1)^2}}. \nonumber 
%	\end{align}
\end{proof}
\newerest{Notice here that $e^\eps-1=O(\eps)$ when $\eps = O(1)$. Hence we have $\expectation{\norm{\widehat{\p} - \p}{1}} = O\left(\sqrt{\frac{k^2}{n\eps^2}}\right)$, which is order optimal.}

%\textcolor{red}{SHOULD WE MENTION SOMWEHERE THAT $e^\eps-1=O(1)$??}

%% file: sec-symmetric.tex
{We show that any private-coin symmetric distribution estimation scheme must communicate at least $\log k$ bits.}

\begin{theorem} \label{thm:dist_sym}
{For any symmetric private-coin scheme without shared randomness that transmits $\ell<\log k$ bits per user, there exists a distribution $\p_0 \in \triangle_k$ such that for $\Xon\sim\p_0$, 
	\begin{equation} \label{eqn:dist_sym}
		\EE\left[ \norm{\hat{\p} (Y^n)- \p_0}{1}\right] \ge 1,
	\end{equation}
	where $Y^\ns$ are the messages sent to $\cR$ after privatizing $\Xon$. }
\end{theorem}
\begin{proof}
%\checked{Let $m:=2^\ell<k$, and wlog assume that $\cY=[m]$, the output alphabet. Then any $\ell$-bit communication scheme that takes as input an element from $[k]$, and outputs a symbol in $[m]$, can be described as a transition probability matrix (TPM) $W\in\RR^{k \times m}$, where
%	\[
%		W(x,y) := \probof{Y = y| X = x}.
%	\]
%	When the input distribution is $\p$, the distribution of the output message is $\q = W^T \p$. Let $\bar{W_i}\RR^{k \times m}$ denote the expected transition matrix for user $i$, where the expectation is over the distribution of $U_i$'s. Since we consider symmetric schemes, let $\bar{W}=\bar{W_1}=\ldots=\bar{W_\ns}$.}
{Assume that $\cY=[2^\ell]$ is the output alphabet. By Lemma~\ref{lem:deterministic}, and symmetry, let $W$ be an $\ell$-bit communication channel used by each user. We can describe $W$ as a transition probability matrix (TPM) $W\in\RR^{k \times 2^\ell}$:
	\[
	W(x,y) := \probof{Y = y| X = x}.
	\]
	When the input distribution is $\p$, the distribution of the output message is $\q = W^T \p$. Notice that $W^T$ is an $2^\ell \times k$ matrix, which is underdetermined since $2^\ell<k$. Therefore, there exists a non-zero vector $\mathbf{e}$ such that $W^T \mathbf{e} = 0$. Further, since $W$ is a TPM, each row of $W$ sums to one, and therefore $W^T \mathbf{e} = 0$ implies that $\sum_{x = 1}^{k} \mathbf{e}(x) = 0$.}

%\newerest{By linearity of expectations,  the distribution of $Y_i$ is $\bar{W}^T\p$. Notice that $\bar{W}^T$ is an $m \times k$ matrix, which is underdetermined since $m<k$. Therefore, there exists a non-zero vector $\mathbf{e}$ such that $W^T \mathbf{e} = 0$. Further, since $W$ is a TPM, 
%	\[
%		\sum_{x = 1}^{k} \mathbf{e}(x) = 0.
%	\]}

By scaling appropriately, we can ensure that $\|\mathbf{e} \|_1 = 2$, which ensures that the positive entries sum to one, and negative entries sum to $-1$. Now consider the distributions specified by these entries, namely let $\p_1 = \max\{\mathbf{e}, 0\}$ and $\p_2 = \max\{-\mathbf{e}, 0\}$. Then these two distributions have \emph{disjoint support}, however, 
	\[
		W^T \p_1 = \q = W^T \p_2,
	\]
	showing that their output message distributions are identical and  they cannot be distinguished. Since $\norm{\p_1- \p_2}{1}= 2$, when we get $Y^n \sim q^n$, for at least one of these distributions, the  expected $\ell_1$ error is 1, proving the result. 
%Let $Y^n$ be $n$ independent samples from $\q$. Then we have:
%	\[
%		\expectation{\dli{\hat{\p} (Y^n)}{\p_1} }+ \expectation{\dli{\hat{\p} (Y^n)}{\p_2} } =  \expectation{\dli{\hat{\p} (Y^n)}{\p_1}  + \dli{\hat{\p} (Y^n}{\p_2}  } \ge  \expectation{\dli{\p_1}{\p_2} } = 1
%	\]
%	And we know one of $\p_1$ and $\p_2$ will satisfy Equation~\ref{eqn:dist_sym}.
\end{proof}

\checked{Note that Theorem~\ref*{thm:dist_sym} holds for all symmetric schemes, not just $\eps$-LDP schemes, which means the result also extends to non-private setting, proving the importance of asymmetry in communication efficient distribution estimation. Further, with just two more bits, using $\log k+2$ bits, HR is private-coin, symmetric, and does optimal distribution estimation.}  % In Section~\ref{sec:hr_heavyhitter}, we will show that, $O(\log k)$ bits of communication is also enough for a symmetric scheme to achieve optimal performance for heavy hitter detection problem using the Hadamard Response scheme proposed in~\cite{AcharyaSZ18a}.

%% file: sec-hadamard-heavyhitter.tex
We first describe the scheme briefly, and prove the optimality. We refer the reader to~\cite{AcharyaSZ18a} for details. Recall that $K$ is the smallest power of 2 larger than $k$. Let $H_K$ be the $K\times K$ Hadamard matrix. The output alphabet of the messages is $\cY:=[K]$ %For $i=1, 2, \ldots, K$ let $B_i\subseteq[K]$ denote the column indices that have a `1' in the $i$th row. Note here $B_1 = [K]$. 
For each input symbol $x\in\{1, \ldots, k\}$, let $C_x$ be the symbols $y\in[K]$ such that there is a 1 in the $y$th column of the $(x+1)$th row of $H_K$. The reason we start with the second row is because the first row of $H_K$ is all one's. 
%Consider the sets $B_1, \ldots, B_{K}$ from the previous section. For each $x\in[k]$, associate $C_x:=B_{x+1}$. 
Since $H_K$ is Hadamard, 
\begin{enumerate}
	\item $\forall x \in [k], |C_x| = \frac{K}{2}$, and 
	\item $\forall x \neq x' \in [k], |C_x \cap C_{x'}| = \frac{K}{4}$.
\end{enumerate}
HR is the following symmetric privatization scheme for all user with output $y\in[K]$, $x\in[k]$,
\begin{equation} \label{eqn:hadamard_response}
	W(x,y) = 
	\begin{cases}
		\frac{2 e^\eps }{ K (1 + e^\eps) } \text{ if } y \in C_x,\\
		\frac{2  }{ K (1 + e^\eps) } \text{ otherwise.}
	\end{cases}
\end{equation}

Consider an arbitrary input $X^\ns$, with $\Mlt{x}$ being the number of appearances of  $x$'s in $X^n$. Let $\Mlt{C_x} :=  \sum_{i \in [n]}  \mathbf{1} \{ Y_i \in C_x \}$ be the number of output symbols that are in $C_x$. Then, we have
\begin{align} 
	\expectation{\Mlt{C_x} } & = \sum_{i \in [n]} \expectation{ \mathbf{1} \{ Y_i \in C_x \} } = \sum_{i \in [n]} \probof{ Y_i \in C_x } \nonumber  \\
	& = \sum_{i \in [n]} \Paren{ \mathbf{1} \{ X_i = x \} \frac{ e^\eps }{ 1 + e^\eps } + \mathbf{1} \{ X_i \neq x \}  \frac{ 1 }{ 2 } } \nonumber \\
	& = \frac{e^\eps -1}{2(e^\eps + 1)}\Mlt{x} + \frac{n}{2}.\label{eqn:unbiased}
\end{align}
Hence,
\begin{equation} \label{eqn:estimator}
	\hat{\p}(x) = \frac{2(e^\eps + 1)}{n(e^\eps -1)} \Paren{\Mlt{C_x} - \frac{n}{2}}.
\end{equation}
is an unbiased estimator for $\frac{\Mlt{x}}{n}$. The performance of the estimator is stated in Theorem~\ref{thm:hadamard_heavy}.

\begin{theorem} \label{thm:hadamard_heavy}
	For any dataset $X^n$, the encoding scheme in~\eqref{eqn:hadamard_response} combined with the estimation scheme in~\eqref{eqn:estimator} satisfies that:
	\begin{equation} \label{eqn:err_infty}
	\expectation{\max_{x \in [k]} \absv{ \hat{\p}(x) - \frac{\Mlt{x}}{n}} } \le \frac{4(e^\eps+1)}{e^\eps-1} \sqrt{\frac{\log k}{n}}.
	\end{equation}
\end{theorem}

\begin{proof}
By~\eqref{eqn:unbiased}, we know the estimator is unbiased. Since each $\Mlt{C_x}$ is a sum of $n$ independent Bernoulli random variables, $\hat{\p}(x)$'s are sub-Gaussian with varaince proxy $\frac{4(e^\epsilon+1)^2}{n(e^\epsilon-1)^2}$. Hence, by Theorem 1.16 from~\cite{Rigollet15}, we get the result in~\eqref{eqn:err_infty}.
\end{proof}

In~\cite{BassilyS15}, a matching lower bound of $\Omega\left(\frac{1}{\eps} \sqrt{\frac{\log k}{n}}\right)$ when $\eps = O(1)$ is proved for LDP heavy hitter estimation algorithms. The above theorem shows that the proposed algorithm has optimal performance. We remark that this scheme has communication complexity of $\log k$ bits per user, and the total computation complexity is $O(k\log k+n)$. The dependence on $k$ is usually undesirable in this problem, and therefore more sophisticated schemes are designed, which require higher communication complexity or shared randomness.

%% file: sec-constantbit.tex
\newest{The previous section showed that with $\ell=\log k+2$ bits of communication per user we can solve heavy hitters problem optimally. In this section, we assume that $\ell<\log k -2$, and prove that there is no private-coin heavy hitter detection scheme that communicates $o(\log n)$ bits per user and is optimal.}

\begin{theorem} \label{thm:heavy_accuracy}
	Let $\ell<\log k-2$. For all private-coin response schemes $(\{W_i\}_{i=1}^n, \hat{\p})$ with only private randomness and $\ell$ bits of communication, there exists a dataset $X_1,\ldots, X_\ns$ with $n > 12(2^\ell +1)^2$, and $x_0\in[k]$ such that: 
	\begin{equation*} \label{eqn:heavy_accruacy}
		\expectation{ \left\lVert{\hat{\p}(\Yon)(x_0) - \frac{\Mlt{x}(X^n)}{n}}\right\rVert_{\infty}} \ge \frac{1}{2^{\ell + 2} +4} ,
	\end{equation*}
	where $Y_i = W_i(X_i)$ for $i \in [n]$.
\end{theorem}
\begin{proof}
We will use the probabilistic method to show the existence of such a dataset. To do so, we design a dataset generating process, and show that the expected $\ell_\infty$ loss over the process and randomness induced by the channels is large, which shows that the expected $\ell_\infty$ loss for the worst dataset is also large. 

Similar to Section~\ref{sec:lower_symmetric}, recall that each $W_i$ can be represented by a $k \times 2^\ell$ transition probability matrix (TPM) where for user $i$, $W_i(x,y) = \probof{Y_i = y| X_i = x}$. Consider distributions $\p_1, \ldots, \p_\ns$ over $[\ab]$, and suppose the data at user $i$, $X_i$ is generated from $\p_i$. Then $\q_i$,  the output distribution of $Y_i$ is given by $W_i^T \p_i$. We will restrict to distributions $\p_i$'s to have support over the first $2^\ell+1$ symbols. Namely, for all $2^\ell+1<x\le \ab$, $\p_i(x)=0$.  Similar to the proof of Theorem~\ref{thm:dist_sym}, since the output dimension is  $2^\ell$, for each $i$, there exists a non-zero vector $\mathbf{e}_i\in\RR^k$, such that $\mathbf{e}_i(x)=0$ for $2^\ell+1<x$, and $W_i^T \mathbf{e}_i = 0$. Further, recall that since $W_i$ is a TPM, $\sum_{x = 1}^{k} \mathbf{e}_i(x) = 0$. Therefore, upon normalizing, assume $\|e_i\|_1 = 2$. 
%	\begin{align}
%		W_i^T \mathbf{e}_i = 0, \sum_{x = 1}^{k} \mathbf{e}_i(x) = 0. \nonumber 
%	\end{align}
%	Without loss of generality, 
Let
	\begin{align}
		\p_i = \max \{\mathbf{e}_i, 0\}, \p'_i = \max \{-\mathbf{e}_i, 0\}. \nonumber 
	\end{align}
	Then $\p_i$ and $\p_i'$ are valid distributions over $[k]$ and effective support only $\{1, \ldots, 2^\ell+1\}$, and $\norm{\p_i - \p_i'}{1} = 2$. Similarly construct $\p_i, \p_i'$ for each $i=1, \ldots, \ns$. Then,
	\begin{align}
		2n= \sum_{x = 1}^{k} \sum_{i = 1}^n |\p_i(x) - \p'_i(x)| = \sum_{x = 1}^{2^\ell+1} \sum_{i = 1}^n |\p_i(x) - \p'_i(x)|, \nonumber 
	\end{align}
%	\begin{align}
%	2n=\sum_{i = 1}^{n} \norm{\p_i - \p_i'}{1} &= \sum_{x = 1}^{k} \sum_{i = 1}^n |\p_i(x) - \p'_i(x)|\nonumber\\
%	&= \sum_{x = 1}^{2^\ell+1} \sum_{i = 1}^n |\p_i(x) - \p'_i(x)|, \nonumber 
%\end{align}
	where we use that $\p_i$, and $\p_i'$ are supported only over the first $2^\ell+1$ symbols. Hence there exists $x_0 \in [2^\ell+1]$, such that
	\begin{equation} 
		\sum_{i = 1}^n |\p_i(x_0) - \p'_i(x_0)| \ge \frac{2n}{2^\ell +1 }. \nonumber 
	\end{equation}
	Without loss of generality, assume $\forall i, \p_i(x_0) \le \p'_i(x_0)$. Then the above equation becomes
	\begin{align} \label{eqn:difference}
		\sum_{i = 1}^n \p'_i(x_0) - \sum_{i = 1}^n \p_i(x_0) \ge \frac{2n}{2^\ell +1 }.
	\end{align}
	
Now consider two datasets generated as follows. $X^n$ satisfies $\forall i \in [n], X_i \sim \p_i$ and $X'^{n}$ satisfies $\forall i \in [n], X_i' \sim \p_i$. Moreover since
	\begin{align}
		W_i^T \p_i = \q_i = W_i^T \p_i', \nonumber 
	\end{align}
	the output distribution $Y^\ns$ is identical for $X'^{n}$, and $X^{n}$. 
	
	Let $N_{x_0}(X^\ns)$ and $N_{x_0}(X'^\ns)$ be the number of appearances of $x_0$ in $X^\ns$ and $X'^\ns$. Then by~\eqref{eqn:difference}, 
	\[
	\expectation{N_{x_0}(X^\ns)}-\expectation{N_{x_0}(X'^\ns)}> \frac{2n}{2^\ell +1}.
	\]
	 Moreover, since $N_{x_0}$ are sum of independent binary random variables, $\variance{N_{x_0}}\le n/4$.  Now suppose $\ell<\frac14\log\ns-1$, then $n/(2^\ell+1)>n^{3/4}$. Therefore, by  Chebychev's inequality, for large $n$, 
	 \[
	 \probof{N_{x_0}(X^\ns)-N_{x_0}(X'^\ns)>\frac{n}{2^\ell +1}}>0.9.
	 \]
	Since the two output distributions are indistinguishable, we have the error is at least $\frac{n}{2^{\ell +1}+2}$ for one of the cases if this event happens. Hence the expectated loss would be at least $0.9 \times \frac{n}{2^{\ell +1}+2} > \frac{n}{2^{\ell +2}+4}$.
\end{proof}
\newerest{Hence we can see when $\ell = O(1)$. We cannot learn the frequency reliably up to accuracy better than a constant. Moreover, when $\ell = o(\log n + \log (1/\eps))$, we get 
\[
	\frac{1}{2^{\ell+1} + 1} > \sqrt{\frac{\log k }{n\eps^2}},
\]
implying that optimal frequency estimation algorithms must require $\Omega(\log n + \log (1/\eps))$ bits of communication when there is no public randomness. Similar to Section~\ref{sec:lower_symmetric}, the result also extends to non-private settings.}
%\begin{theorem} \label{thm:heavy_accuracy}
%	For all response schemes $(\{f_i\}_{i=1}^n, \hat{\p})$ with only private randomness and constant bits of communication, there exist one possibly random dataset $\{X_i\}_{i=1}^{n}$ and a constant $c_0$, such that: 
%	\begin{equation}
%		\probof{ \exists x \in [k], |\hat{\p}(x) - \frac{M_x(X^n)}{n}| \ge c_0} \ge \frac{1}{2}
%	\end{equation}
%\end{theorem}

%\begin{theorem} \label{thm:heavy_accuracy}
%	For all response schemes $\{\cR_i\}_{i=1}^n$ with only private randomness and less than $\log k$ communication, there exist two possibly random datasets $\{X_i\}_{i=1}^{n}, \{X'_i\}_{i=1}^{n} \in [k]^n$ such that:
%	\begin{enumerate}
%		\item $\forall i \in [n]$, $\cR_i(X_i)$ and $\cR_i(X'_i)$ have the same distribution.
%		\item There exists $x \in [k]$, $|M_x(X_i) - M_x(X'_i)| \ge \frac{n}{k}$ with probability one. 
%	\end{enumerate}
%\end{theorem}

%% file: experiments.tex
We conduct empirical evaluations for the one-bit distribution learning algorithm without public randomness proposed in Section~\ref{sec:onebit}. We compare the proposed algorithm (onebit) with other algorithms including Randomized Response (RR)~\cite{Warner65}, RAPPOR~\cite{ErlingssonPK14}, Hadamard Response (HR) ~\cite{AcharyaSZ18a} and subset selection (subset)~\cite{YeB18}. To obtain samples, we generate synthetic data from various classes of distributions including uniform distribution, geometric distributions with parameter 0.8 and 0.98, Zipf distributions with parameter 1.0 and 0.5 and Two-step distribution. We conduct the experiments for $k = 1000$ and $\eps =1$. The results are shown in Figure~\ref{fig:k-1000}. %For $k = 10000$ and for $\p \sim Geo(0.8)$, the result is shown in Figure~\ref{fig:k-10000}. 
Each point is the average of 30 independent experiments.

From the figures, we can see the performance of our proposed scheme is comparable to the best among all schemes for various kinds of distributions. And the communication complexity is only one bit while the least among others is $\Omega(\log k)$ bits~\cite{AcharyaSZ18a}.

\begin{figure} [h]
	\centering
	\subfigure[	{{\small Uniform}} ]{\includegraphics[width=0.22\textwidth]{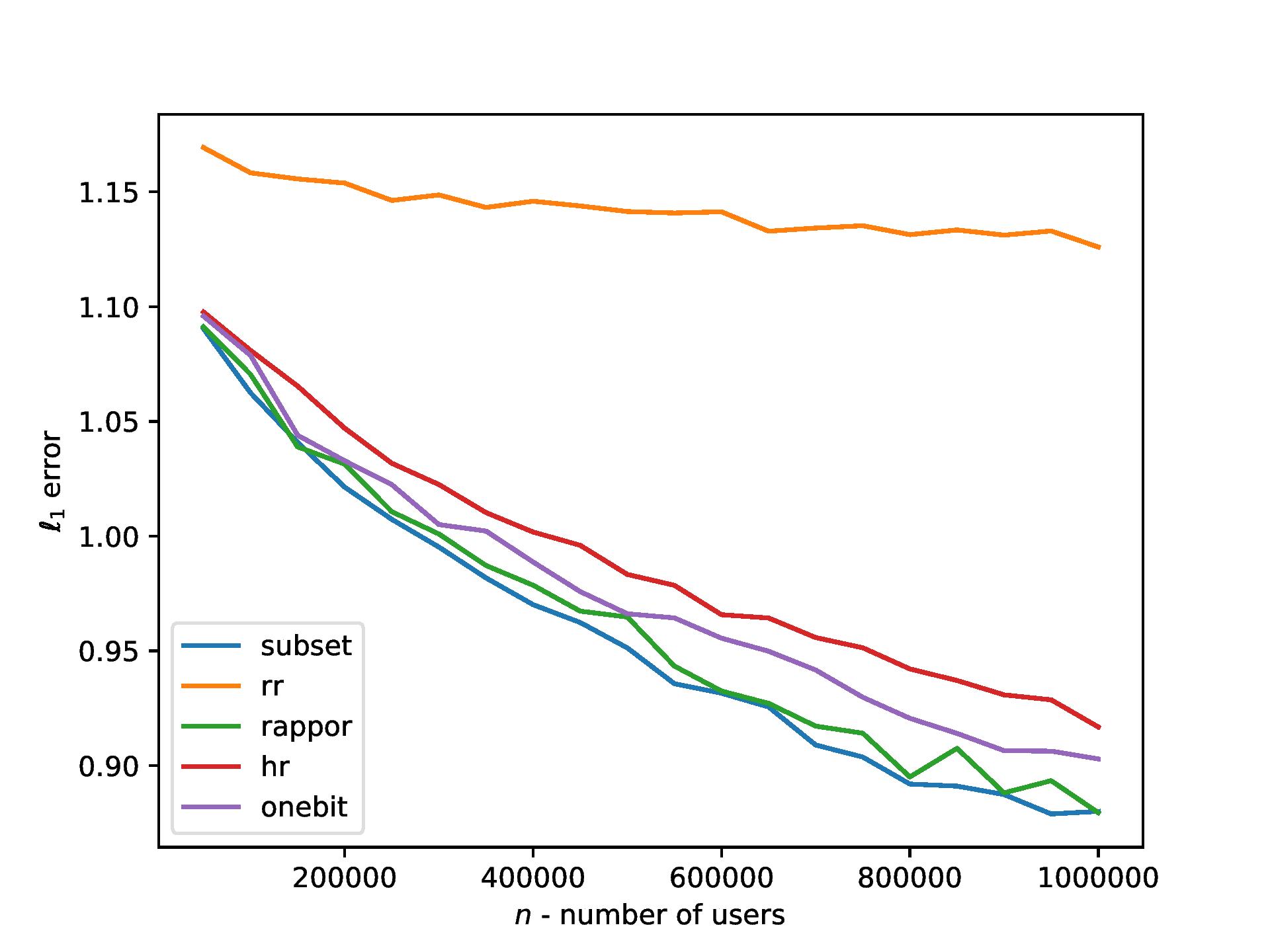}\label{fig:k_1000_eps_05}}
	\subfigure[	{{\small Geo$(0.8)$}} ]{\includegraphics[width=0.22\textwidth]{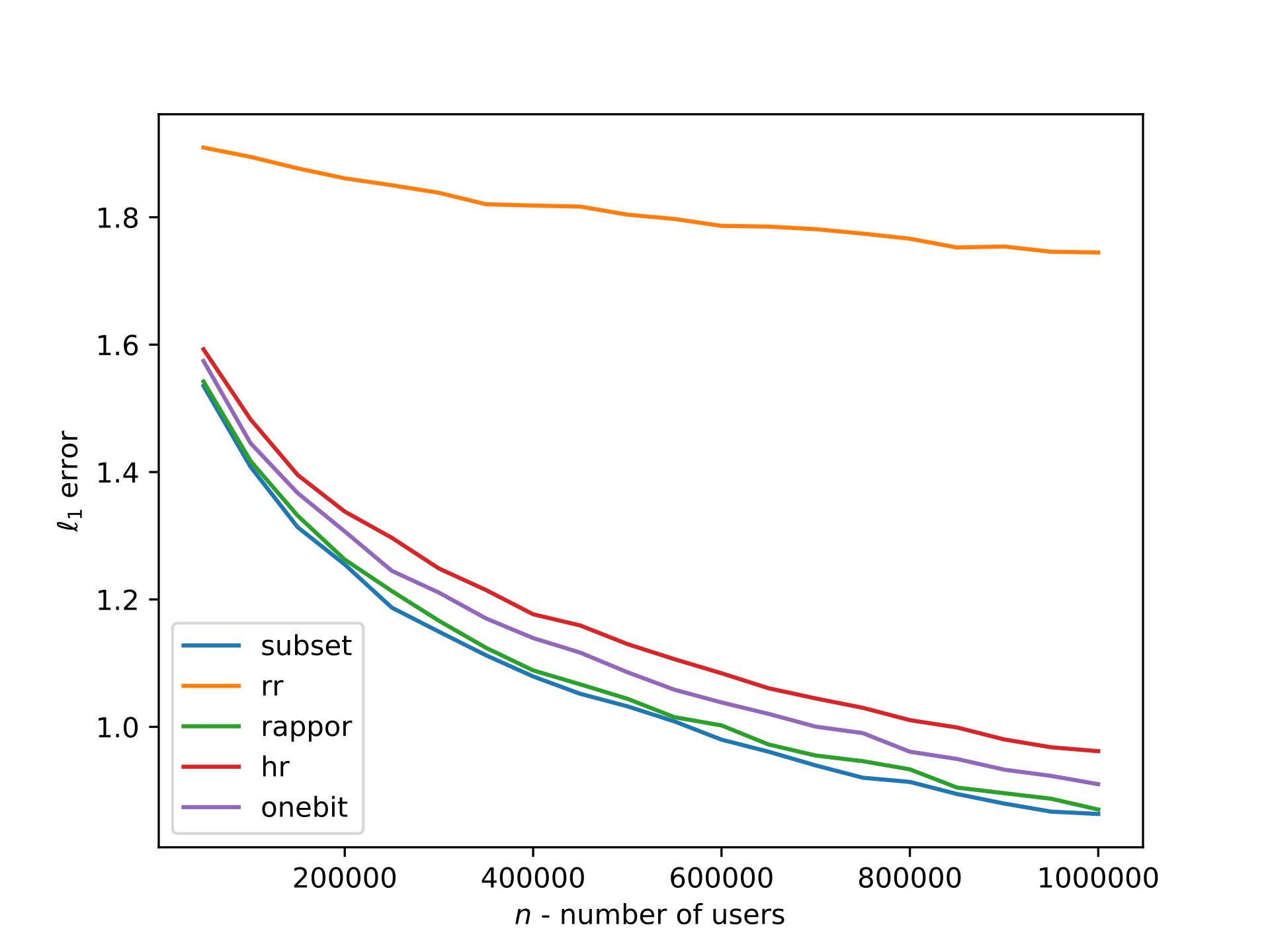}\label{fig:k_1000_eps_2}}
	\subfigure[	{{\small Geo$(0.98)$}} ]{\includegraphics[width=0.22\textwidth]{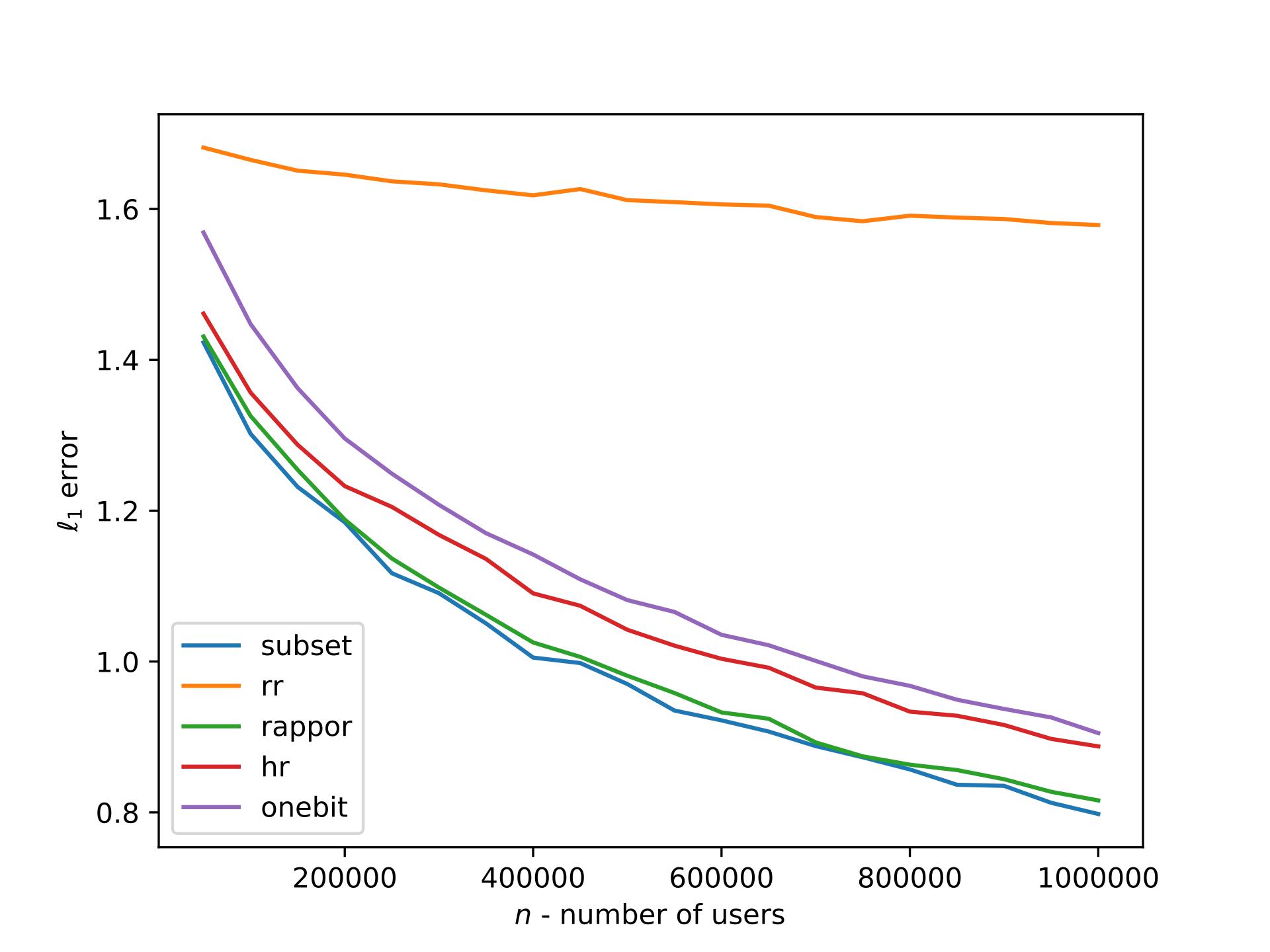}\label{fig:k_1000_eps_7}}	
	\subfigure[	{{\small Zipf$(0.5)$}} ]{\includegraphics[width=0.22\textwidth]{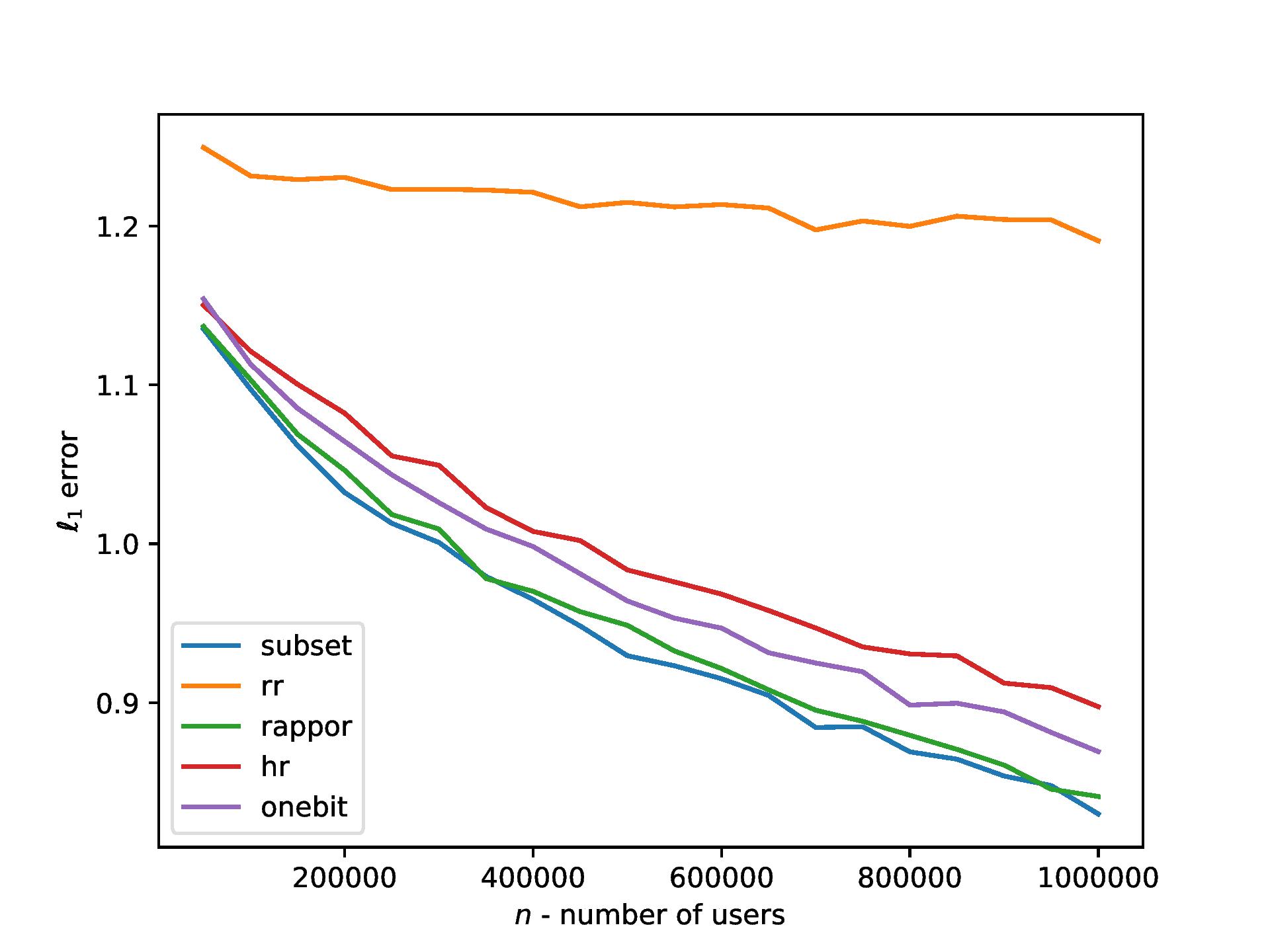}\label{fig:k_1000_eps_5}}
	\subfigure[	{{\small Zipf$(1.0)$}} ]{\includegraphics[width=0.22\textwidth]{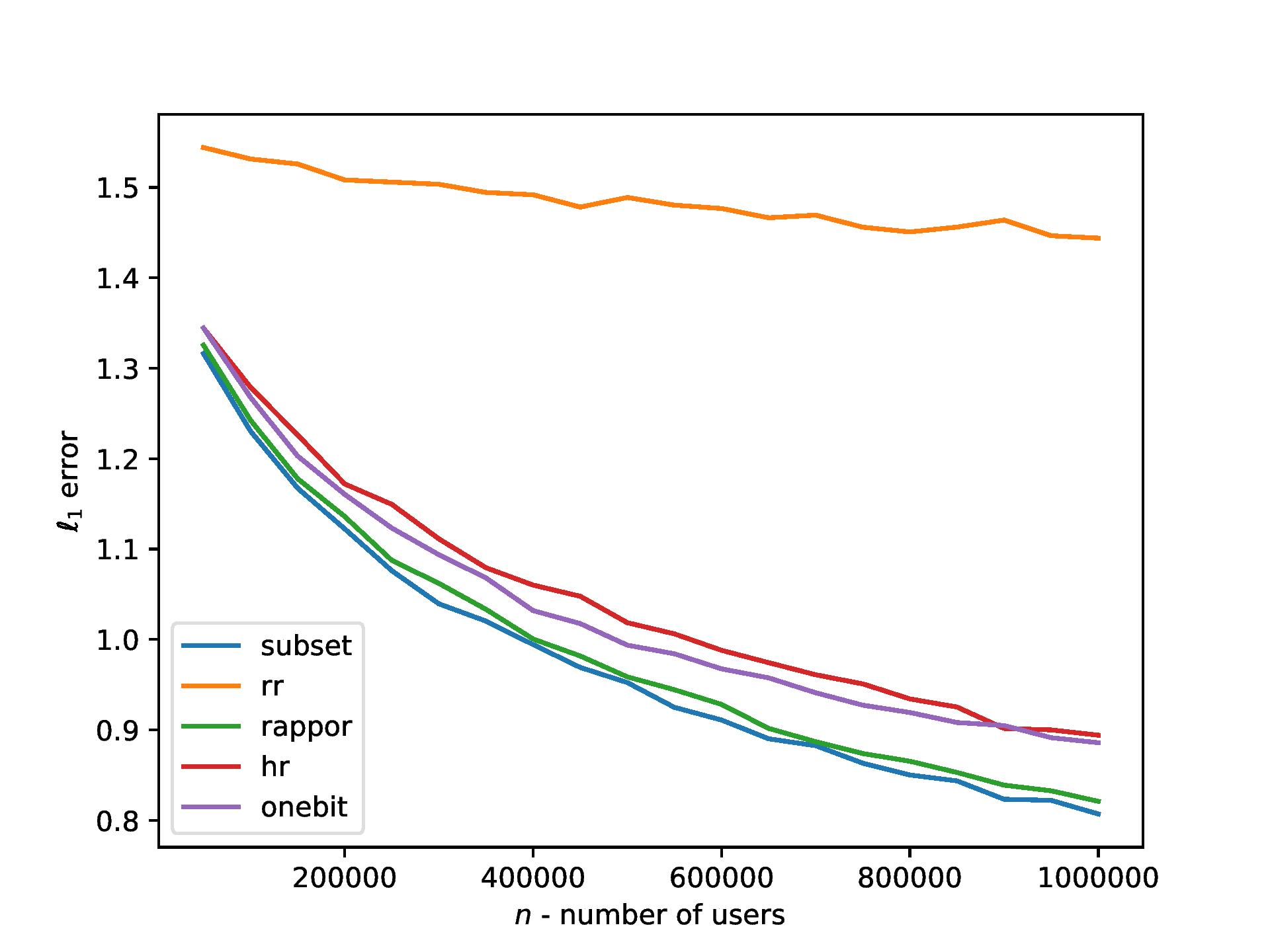}\label{fig:k_1000_eps_7}}
	\subfigure[	{{\small Two Steps}} ]{\includegraphics[width=0.22\textwidth]{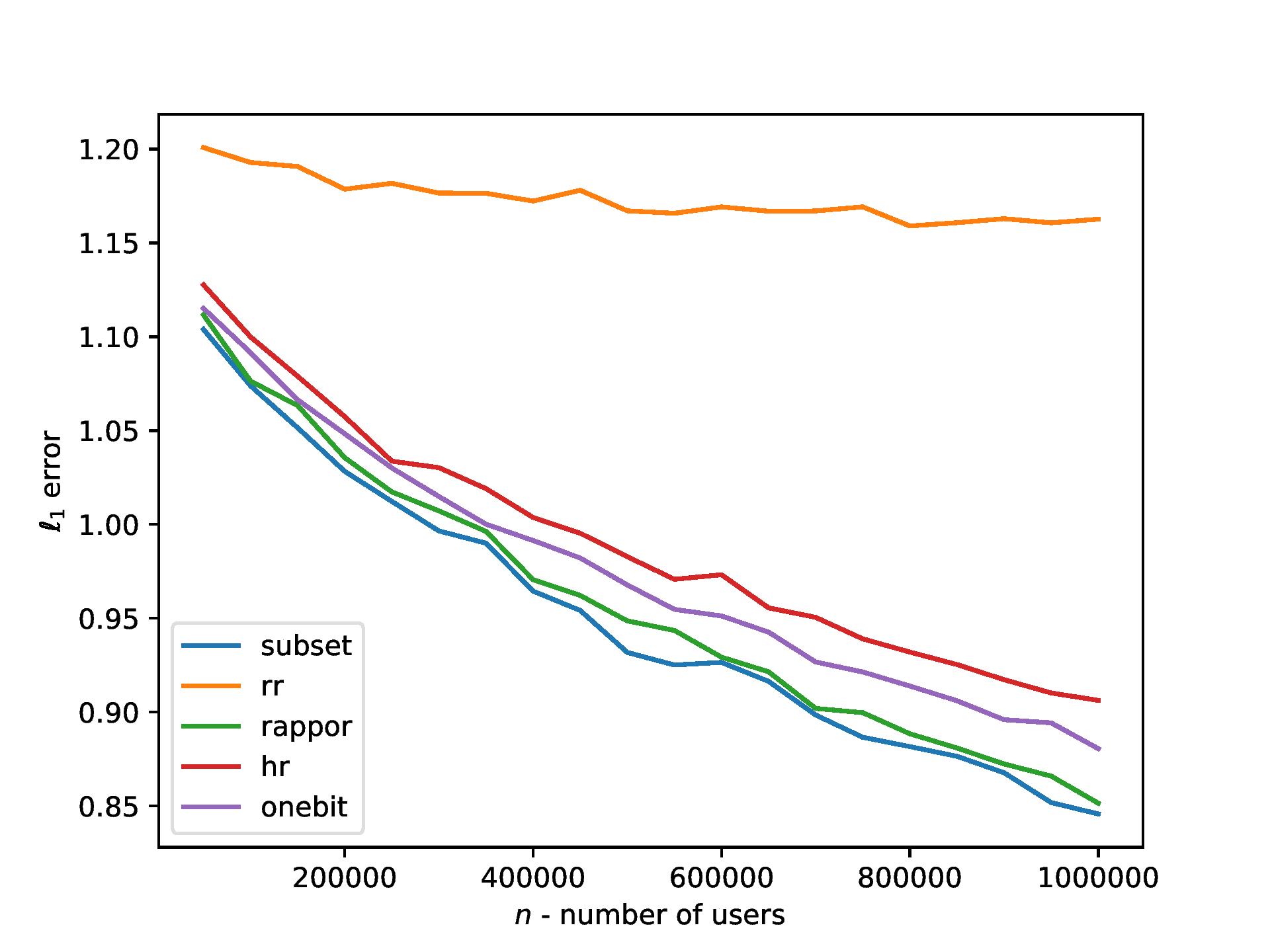}\label{fig:k_1000_eps_5}}
		\caption[ ]
	{\small $\ell_1$-error for $k = 1000$, $p$ from Uniform, Geo$(0.8)$, Geo$(0.98)$, Zipf$(0.5)$, Zipf$(1.0)$ and Two-step distributions.}   
	\label{fig:k-1000}
\end{figure}

%